\documentclass[11pt]{elsarticle}
\pdfoutput=1
\usepackage{setspace}
\usepackage[plain]{fullpage}
\usepackage{algorithm}
\usepackage{algorithmicx}
\usepackage[noend]{algpseudocode}
\usepackage{algpascal}
\usepackage{epsfig}
\usepackage{enumerate}
\usepackage{paralist}
\usepackage{color}
\usepackage{srcltx}
\usepackage{amsmath, amssymb, amsthm}

\usepackage{amsmath}
\usepackage{amssymb}
\usepackage{color}
\definecolor{red}{RGB}{255,0,0}
\definecolor{blue}{RGB}{0,0,255}
\definecolor{green}{RGB}{0,255,0}

\newcommand {\abs}[1]  {\left\vert#1\right\vert}
\newcommand {\set}[1]  {\left\{#1\right\}}

\newcommand {\defined} {\stackrel{def} {=}}
\newcommand {\np}      {\textsc{NP}}
\newcommand {\npc}     {\textsc{NP}\textrm{-complete}}

\newcommand {\bigoh}   {{\cal O}}

\newcommand {\runningtitle}[1] {\vspace{0.5ex}\noindent{\textbf{\boldmath #1:}}}

\newcommand{\ignore}[1] {}
\newcommand {\commentfig}[1] {#1}

\newcommand{\LineComment}[1]{\State~\(\triangleright\)~#1}
\newtheorem{theorem} {Theorem}
\newtheorem{lemma} {Lemma}

\newtheorem{observation}  {Observation} 

\usepackage{tikz}

\newcommand{\optprobl}[4]
{
  \begin{flushleft}
    \fbox{
      \begin{minipage}{\textwidth}
        \noindent {\textsc {#1}}\\
        {\bf Input:} #2\\
        {\bf Output:} #3\\
        {\bf Objective:} #4
      \end{minipage}
    }
  \end{flushleft}
}

\newcommand{\bb}{{\cal B}}

\newcommand{\uu}{{\cal U}}

\newcommand{\kdefprobl}{$k$-\textsc{DefensiveDomination}}
\newcommand{\kplusdefprobl}{$(k+1)$-\textsc{DefensiveDomination}}

\begin{document}
\begin{frontmatter}
\title{Defensive Domination in Proper Interval Graphs}
\author[bu]{T{\i}naz Ekim\fnref{tubitak}}
\ead{tinaz.ekim@boun.edu.tr}

\author[uo]{Arthur Farley}
\ead{art@cs.uoregon.edu}

\author[uo]{Andrzej Proskurowski}
\ead{andrzej@cs.uoregon.edu}

\author[isik]{Mordechai Shalom}
\ead{cmshalom@gmail.com}

\fntext[tubitak]{Supported by TUBITAK Grant no: 118E799.}

\address[bu]{Department of Industrial Engineering, Bo\u{g}azi\c{c}i University, Istanbul, Turkey}
\address[uo]{University of Oregon, Eugene, Oregon 97403 USA}
\address[isik]{I\c{s}{\i}k University, Istanbul, Turkey}

\begin{abstract}
$k$-defensive domination, a variant of the classical domination problem on graphs, 
seeks a minimum cardinality vertex set providing a surjective
defense against any attack on vertices of cardinality bounded by a parameter $k$. 
The problem has been shown to be $\npc$ for fixed $k$;
if $k$ is part of the input, the problem is not even in $\np$. 
We present efficient algorithms solving this problem on proper interval graphs with $k$ part of the input.  
The algorithms take advantage of the linear orderings of the end points of the intervals associated with vertices to realize a greedy approach to solution.  
The first algorithm is based on the interval model and has complexity $\bigoh(n \cdot k)$ for a graph on
$n$ vertices. 
The second one is an improvement of the first and employs bubble representations of proper interval graph to realize an improved complexity of $\bigoh(n+ \abs{\bb} \cdot \log k)$ for a graph represented by $\abs{\bb}$ bubbles.
\end{abstract}

\end{frontmatter}

\section{Introduction}\label{sec:intro}
In the classical domination problem, we determine a minimum cardinality set of vertices such that every other vertex has a neighbor in this set. 
Various applications can be modeled as variations of the classical domination problem. 
Due to the wide range of applications and interesting theoretical properties, domination problems have been extensively studied in the literature (see e.g., \cite{HH1,HH}). 

One application area for domination problems is  network security. 
Indeed, a vertex in a dominating set is often considered to be a security guard placed at that vertex having the ability to protect itself and its neighbors (i.e., vertices having a direct link to it) in the network. 
The defensive domination problem, introduced in \cite{FP}, is a variation that deals with the capability of a network to be self-organized when several vertices are under attack simultaneously. 
In the $k$-defensive domination problem, any subset of vertices of cardinality up to a given number $k$ is seen as a potential attack against which a defensive domination set must defend. 

Despite its close links with other domination problems and network security problems, the literature on the defensive domination problem is limited. 
In \cite{FP}, the authors introduce the defensive domination problem along with its fundamental properties, and a dynamic programming algorithm to solve the problem in trees is presented. 
In a related work on defensive domination \cite{kra}, the authors use a slightly different version of the defensive domination problem as a tool to solve the cops-and-robber (or pursuit-evasion) game in interval graphs. 
More recently, in \cite{EFP}, the authors consider the computational complexity of the defensive domination problem. They first show that the defensive domination problem is $\npc$, even for split graphs, if the maximum cardinality $k$ of an attack is fixed. 
They also show that if $k$ is not fixed, the problem is not even in $\np$. Indeed, if $k$ is not fixed, there are exponentially many potential attacks (subsets of $k$ vertices) to be considered. 
Subsequently, the authors describe optimal $k$-defensive dominating sets for cycles and paths (solving the problem in linear time) and develop linear time algorithms for co-chain graphs and threshold graphs, even when $k$ is not fixed. 
In \cite{EFP}, the authors point out the complexity of the defensive domination problem in proper interval graphs as an open question. 
In the present paper, we show that the defensive domination problem can be solved in polynomial time in proper interval graphs, generalizing the result for co-chain graphs. 

In Section \ref{sec:prelim}, we provide formal definitions for the $k$-defensive domination problem and proper interval graphs along with some general observations about defensive dominating sets. Section \ref{sec:properIntervalGraph} describes a greedy algorithm that finds a minimum $k$-defensive dominating set in proper interval graphs in time $\bigoh(n \cdot k)$. Our approach is similar to the ones used in \cite{EFP}. Namely, we take advantage of the  structure of proper interval graphs to show that if a set of vertices can defend against a selected set of $k$-attacks then it can defend against all $k$-attacks, ensuring that it is a $k$-defensive dominating set. In Section \ref{sec:bubblesModel}, we improve the time complexity of our greedy algorithm to $\bigoh(n+ \abs{\bb} \cdot \log k)$ where $\abs{\bb}$ stands for the number of bubbles in a compact bubble representation of proper interval graphs.

\section{Preliminaries}\label{sec:prelim}
\runningtitle{Graph notations and terms}
Given a simple graph (no loops or parallel edges) $G=(V,E)$, $uv$ denotes an edge between two vertices $u,v$ of $G$. 
We denote by $N(v)$ the set of neighbors of $v$ in $G$, and by $N[v]$ the \emph{closed neighborhood} of $G$, i.e. $N(v) \cup \set{v}$.
We extend this notion to a set of vertices in the standard way: for $A \subseteq V$, $N[A]=\bigcup_{v \in A} N[v]$.
Two adjacent vertices $u,v$ of $G$ are \emph{twins} if $N(u) \setminus \set{v} = N(v) \setminus \set{u}$. For a graph $G$ and $U \subseteq V(G)$, we denote by $G[U]$ the subgraph of $G$ induced by $U$.

The square $G^2$ of a graph $G$ is a graph over the same vertex set (i.e., $V$),
such that two vertices $u$ and $v$ are adjacent in $G^2$ if and only if their distance is at most 2 in $G$. We denote by $G^2[A]$ the subgraph of $G^2$ induced by $A\subseteq V(G)$.

A graph $G$ is an \emph{interval graph} if its vertices can be represented by closed intervals on the real line, such that two vertices are adjacent in $G$ if and only if the corresponding intervals are intersecting. 
An interval graph is \emph{proper}  if it has an interval representation such that no interval properly contains another. 
If all intervals of its interval representation have the same length, then the graph is called a \emph{unit} interval graph. 
It is known that the class of proper interval graphs is equivalent to the class of unit interval graphs \cite{Bogart199921}.
It is important to note that in a proper interval representation of a proper interval graph, the left-to-right order of the left endpoints of intervals is the same as the order of their right endpoints (See Figure \ref{fig:ProperIntervalRep1}).
\begin{figure}
\begin{center}
\scalebox{1}{\commentfig{
\includegraphics[width=11cm]{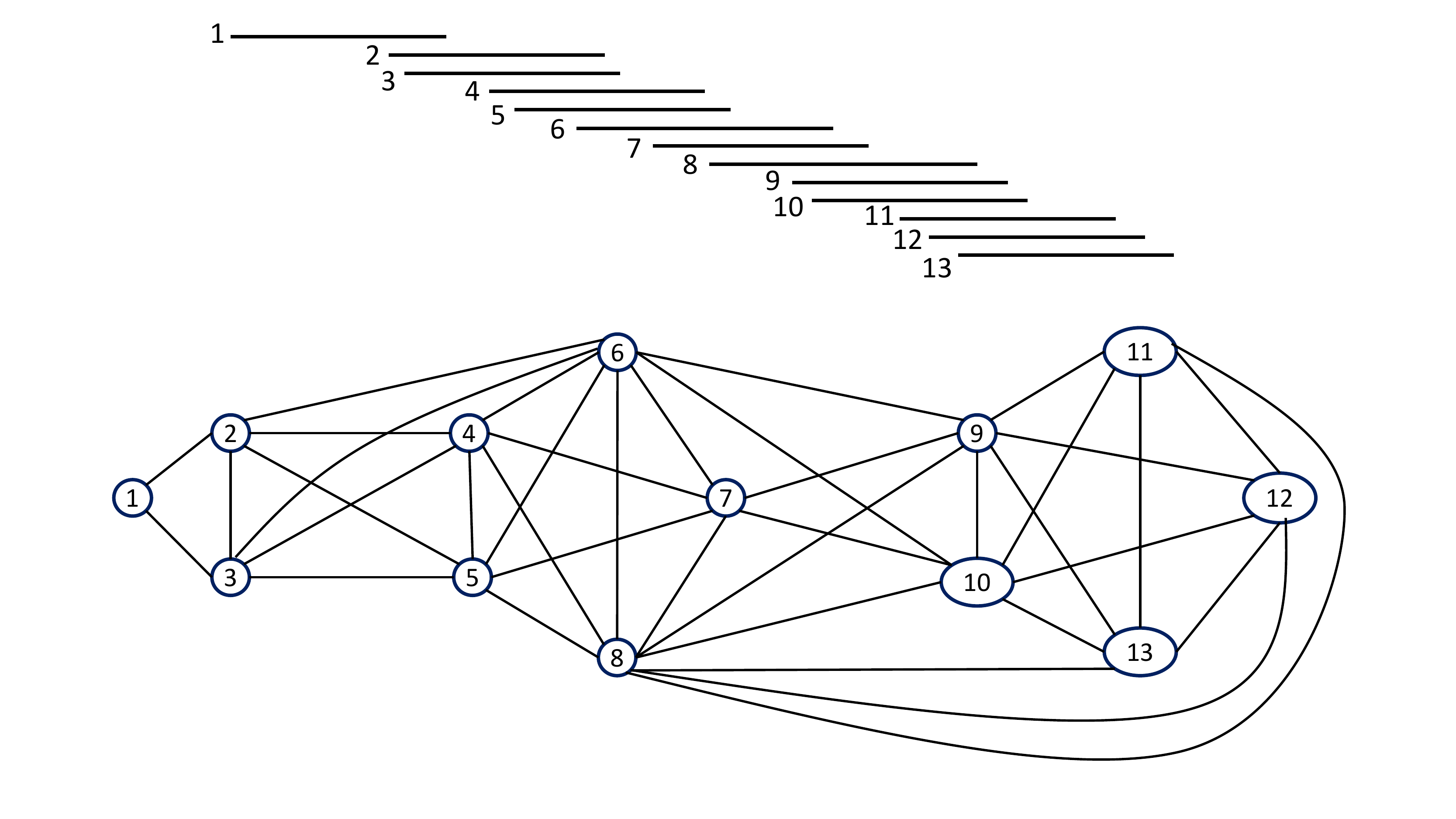}}}
\caption{A proper interval representation and the corresponding proper interval graph.}\label{fig:ProperIntervalRep1}
\end{center}
\end{figure}

\runningtitle{$k$-Defensive Domination}
An \emph{attack} on a graph is a subset of its vertices, 
and an element (i.e., a vertex) of an attack is termed an \emph{attacker}.
A $k$-\emph{attack} is an attack with at most $k$ vertices.
Given a graph $G$, an attack $A$ on $G$, and a set $D \subseteq V(G)$ of \emph{defenders},
a (possibly partial) surjection $f: D \rightarrow A$ is termed a \emph{defense} of $D$ against $A$ if, for $d \in D$, $f(d) \in N[d]$ whenever $f(d)$ is defined.
A set $D$ of vertices \emph{defends} against $A$ if there exists a defense of $D$ against $A$. 
A set $D$ of vertices of $G$ is \emph{$k$-defensive} if it defends against any $k$-attack on $G$.
Clearly, a $k$-defensive set always exists, since the set $V$ defends against any attack.

In this work we consider the following optimization problem.
\optprobl{\kdefprobl}
{A graph $G$, and a positive integer $k$.}
{A $k$-defensive set $D$ of vertices of $G$.}
{Minimize $\abs{D}$.}
We assume that $0<k<\abs{V}$, since otherwise the problem is trivial.
The following two observations are valid for general graphs. 

\begin{observation}\label{obs:ConnectedInGTwo}
A set $D$ of vertices is $k$-defensive if and only if it defends against every $k$-attack $A$ such that $G^2[A]$ is connected.
\end{observation}
\begin{proof}
Necessity being obvious, we prove sufficiency. 
Let $D$ be a set of vertices that defends against all attacks $A$ such that $G^2[A]$ is connected.
Suppose that $D$ does not defend against a $k$-attack $A'$. 
Then $G^2[A']$ consists of connected components $G^2[A'_1], \ldots, G^2[A'_\ell], \ell > 1$.
Therefore, $N[A'_i] \cap N[A'_j] = \emptyset$ for every distinct $i,j \in [1..\ell]$.
Let $D_i \defined D \cap N[A'_i]$ for every $i \in [\ell]$.
Clearly, $\set{D_1, \ldots D_\ell}$ is a near partition of $D$, possibly including empty sets.
By our assumption, there exists $i \in [\ell]$ such that $D_i$ does not defend against $A'_i$.
Then $D$ does not defend against $A'_i$, and $G^2[A'_i]$ is connected, contradicting our assumption.
\end{proof}

We can extend the above observation using Hall's theorem (about systems of distinct representatives) as follows. 
\begin{observation}\label{obs:Hall}
$D$ is $k$-defensive if and only if $\abs{N[A] \cap D} \geq \abs{A}$ for every set $A$ with at most $k$ vertices such that $G^2[A]$ is connected. 
\end{observation}

\section{A Greedy Algorithm}\label{sec:properIntervalGraph}

Our greedy algorithm is based on the following idea. We first show that defending against a few special attacks is sufficient for a set $D$ to be a $k$-defensive dominating set. These special $k$-attacks can be obtained from the interval representation of a proper interval graph and are the set of consecutive vertex intervals of size $k$. Note that this is well-defined since the intervals of a proper interval graphs are uniquely ordered by the left endpoints (or equivalently the right end-points). Our greedy algorithm considers these $k$-attacks one by one and adds defenders only when necessary to extend an optimal defense.  

Given two natural numbers $i \leq j$ we denote by $[i..j]$ the set of consecutive integers from $i$ to $j$ including $i$ and $j$.
In this section, unless stated otherwise, $G=(V,E)$ is a connected proper interval graph with $V=[1..n]$, where each vertex $i \in V$ is represented by an interval $I_i$ on the real line, 
where no interval $I_i$ is properly included in another, 
and $ij \in E$ if and only if $I_i \cap I_j \neq \emptyset$ for every two distinct $i,j \in [1..n]$.
The vertices are numbered according to the left-to-right order of the left endpoints of the corresponding 
intervals (which is equal to the order of their right endpoints since $G$ is a proper interval graph).

For a subset $A \subseteq V$ of vertices of $G$, denote by $\uu(A)=\bigcup_{i \in A} I_i$ the union of the intervals representing them.  Note that when two intervals overlap their union is one interval with their combined extents.  Let intervals $I'_1, \ldots, I'_c$ (numbered in their left-to-right order on the real line) be the collection of disjoint intervals of $\uu(A)$.
The \emph{span} of $A$ is the interval on the real line between the left endpoint of $I'_1$ and the right endpoint of $I'_c$.  We define $range(A)$ as the set of vertices of $G$ 
whose representations are included in the span of $A$. 
 Let $\min(A)$ be the leftmost vertex in $I'_1$ and let $\max(A)$ be  the rightmost vertex in $I'_c$.
Thus $range(A) = [\min(A) .. \max(A)]$.
We call a subset $A \subseteq V$ of vertices of $G$ \emph{consecutive} if $A=[i..j]$ for two integers $i,j \in [1..n]$,
i.e., if $A = range(A)$.

We say that $A$ is \emph{bridged} if, 
for every two consecutive intervals $I'_i$ and $I'_{i+1}$ of $\uu(A)$, there is a vertex  $I_j$, termed a \emph{bridging vertex} of $I'_i$, such that $I_j$ intersects both $I'_i$ and $I'_{i+1}$.
Note that in a proper interval graph the bridging vertices of distinct pairs of intervals of $\uu(A)$ are distinct.

We start with the following observation that characterizes the subsets $A$ such that $G^2[A]$ is connected.
\begin{observation}\label{obs:GSquareUConnected}
Let $G$ be a proper interval graph and $A \subseteq V(G)$ a subset of its vertices.
Then $G^2[A]$ is connected if and only if $A$ is bridged. 
\end{observation}

\begin{lemma}\label{lem:rangeNeighbourhood}
Let $G$ be a proper interval graph and $A \subseteq V(G)$ be a subset of its vertices.
If $G^2[A]$ is connected then $N[A] = N[range(A)]$.
\end{lemma}
\begin{proof}
Clearly, $N[A] \subseteq N[range(A)]$ since $A \subseteq range(A)$.
It remains to show that $N[range(A)] \subseteq N[A]$.
Since $G^2[A]$ is connected $A$ is a bridged set by Observation \ref{obs:GSquareUConnected}.
Assume that there exists a vertex $i \in N[range(A)] \setminus N[A]$.
Then $I_i$ is strictly between two consecutive intervals of $\uu(A)$
and therefore properly included in an interval representing a bridging vertex.
This would contradict the premise of the Lemma that $G$ is a proper interval graph.
\end{proof}

After removing the leftmost interval of a bridged set $A$, it remains bridged;
this also holds for the addition of a vertex $v \in range(A)$ to $A$.
Consequently, we have the following observation.
\begin{observation}\label{obs:GSquareRemainsConnected}
Let $G$ be a proper interval graph and $A \subseteq V(G)$ be a subset of its vertices with $u$ being the leftmost vertex of $A$ and $v$ any vertex in $range(A)$.
If $G^2[A]$ is connected then $G^2[A - u + v]$ is connected.
\end{observation}

This allows us to define a small set of $k$-attacks such that defending against those attacks suffices to defend against any $k$-attack.

\begin{lemma}\label{lem:consecutive}
Let $G$ be a proper interval graph, and $D$ be a set of vertices of $G$. 
Then $D$ is $k$-defensive if and only if $D$ defends against every consecutive $k$-attack in $G$.
\end{lemma}
\begin{proof}
Necessity being obvious, we prove sufficiency.
Assume that $D$ defends against any consecutive $k$-attack in $G$,
but there exists a $k$-attack in $G$ that $D$ does not defend against.
By Observation \ref{obs:Hall}, there exists a  $k$-attack $A$ such that
$G^2[A]$ is connected and $\abs{N[A] \cap D} < \abs{A}$.
Let $A$ be such an attack where $\abs{range(A)}$ is the smallest possible.
By our assumption, $A$ is not consecutive, since $D$ does not defend against $A$.
Since $A$ is not consecutive, there exists $\ell \in range(A) \setminus A$.
Let $A'= A - \min A + \ell$. 
By Observation \ref{obs:GSquareRemainsConnected}, $G^2[A']$ is connected.
Clearly, $range(A') \subseteq range(A)$. 
Therefore, using Lemma \ref{lem:rangeNeighbourhood} we have $N[A']=N[range(A')] \subseteq N[range(A)]=N[A]$.
This implies
\[
\abs{N[A'] \cap D} \leq  \abs{N[A] \cap D} < \abs{A} = \abs{A'}.
\]
However, $\abs{range(A')} = \abs{range(A)}-1$, contradicting the minimum size of the range of $A$.
\end{proof}

We now observe a few useful properties of sets that defend against attacks on proper interval graphs.
A defense $f$ of a set $D$ against an attack $A$ on $G$ is \emph{monotonic} if $f(i) < f(j)$ whenever $i < j$. 

\begin{lemma}\label{lem:defenseMonotonic}
Let $G$ be a proper interval graph over the vertex set $[1 .. n]$ and $A$ be an attack in $G$.
Let $D \subseteq [1 .. n]$ be a set that defends against $A$.
Then there exists a monotonic defense of $D$ against $A$. 
\end{lemma}
\begin{proof}
Let $f$ be a defense of $D$ against $A$ with the smallest number of inversions, i.e. smallest numbers of pairs such that $i < j$ but $f(i) > f(j)$.
If the number of inversions is zero, we are done.
Otherwise, consider an inversion of $f$, i.e., a pair $(i,j)$ such that $i < j$ but $f(i) > f(j)$.
Since $f$ is a defense, $f(i) \in N[i]$ and $f(j) \in N[j]$.
We now show that $f(i) \in N[j]$.
Indeed, suppose that $f(i) \notin N[j]$. 
We observe that since $i < j$, $f(i) \in N[i] \setminus N[j]$ and $f(j) \in N[j]$ we have $f(i) < f(j)$, a contradiction.
By symmetry, we have $f(j) \in N[i]$.
Therefore, we can obtain a defense $\bar{f}$ with a smaller number of inversions, 
by setting $\bar{f}(i)=f(j)$, $\bar{f}(j)=f(j)$, and $\bar{f}(\ell)=f(\ell)$ for every $\ell \notin \set{i,j}$.
This contradicts the definition of $f$. 
\end{proof}

\begin{lemma}\label{lem:defenseShift}
Let $G$ be a proper interval graph over the vertex set $[1..n]$ and $A=[i .. j]$ be a consecutive attack in $G$.
Let $\delta$ be a positive integer such that $D=[i+\delta .. j+\delta]$ defends against $A$. 
Then $D'=[i+\delta-1 .. j+\delta-1]$ defends against $A$.
\end{lemma}
\begin{proof}
By Lemma \ref{lem:defenseMonotonic}, there is a monotonic defense $f$ of $D$ against $A$ such that $f(\ell+\delta)=\ell$ for every $\ell \in A$.
We observe that if $\ell+\delta$ is adjacent to $\ell$ then $\ell+\delta-1$ is adjacent to $\ell$.
Therefore, there is a defense ${f'}$ of $D'$ against $A$ such that $f'(\ell+\delta-1)=\ell$ for every $\ell \in A$.
\end{proof}

We are now ready to present our greedy algorithm along with a proof of its correctness and complexity.

\newcommand{\alggreedy}{\textsc{Greedy}}

\alglanguage{pseudocode}

\begin{algorithm}
\caption{\alggreedy}\label{alg:Greedy}
\begin{algorithmic}[1]
\Require {A proper interval graph $G=([1..n],E)$, an integer $k$ less than $n$.}
\Ensure {A $k$-defensive subset $D$ of $[1..n]$ of smallest cardinality.}
\Statex
\State $D \gets \emptyset$.
\For {$j=1$ to $n$}
\State $A \gets [\max \set{1,j-k+1} ..~j]$. \label{lin:CalcA}
\If {$D$ does not defend against $A$}\label{lin:GreedyIf}
\State $j' \gets \max (N[A] \setminus D)$.
\State $D \gets D + j'$.
\EndIf
\EndFor
\State \Return $D$.

\end{algorithmic}
\end{algorithm}

\begin{theorem}\label{thm:Greedy}
Algorithm $\alggreedy$ returns an optimal solution $D$ of {\kdefprobl} in time $\bigoh{(n \cdot k)}$.
\end{theorem}
\begin{proof}
Algorithm {\alggreedy} starts with an empty set $D$ and considers the vertices $j$ of $G$ one at a time in the left-to-right order, i.e. from $1$ to $n$.
It adds vertices to $D$ so as to maintain the invariant defined below, such that at the end of iteration $j$, $D$ defends all $k$-attacks on $V_j=[1..j]$.
Finally, it returns $D$ after all vertices have been processed.

Let $D_j$ be the value of $D$ at the end of iteration $j$, for $j \in [1..n]$, and let $D_0=\emptyset$ (the value of $D$ at the beginning of iteration 1).
In the remainder of the proof, we show that for every $j \in [0..n]$, the following invariant holds:
\begin{enumerate}
\item \label{itm:DefendsVj} $D_j$ defends against every $k$-attack in $V_j$, and
\item \label{itm:ContainedInOPT} there exists an optimal solution $D^*$ of {\kdefprobl} of $G$ that is a superset of $D_j$.
\end{enumerate}

The invariant clearly holds for $j=0$. 
We now assume that the invariant holds for $j-1$ and prove that it holds for $j$.
Since, by the inductive hypothesis, $D_{j-1}$ defends against all $k$-attacks on $V_{j-1}$, and $D_j \supseteq D_{j-1}$,
in order to prove (i),
by Lemma \ref{lem:consecutive} it is sufficient to show that $D_j$ defends against every consecutive $k$-attack on $V_j$ that contains $j$.
In other words, it is sufficient to show that $D_j$ defends against the attack $A$ computed at Line \ref{lin:CalcA} at iteration $j$ of the algorithm.
We consider two cases for iteration $j$.
\begin{itemize}
\item {\textbf{$D=D_{j-1}$ defends against $A$ at Line \ref{lin:GreedyIf}.}}
In this case, by the behaviour of the algorithm we have $D_j = D_{j-1}$.
Then $D_j$ defends against $A$, and thus defends against every consecutive $k$-attack on $V_j$ as required.
Furthermore, by the inductive hypothesis, there exists an optimal solution $D^* \supseteq D_{j-1} = D_j$, which proves (ii).

\item {\textbf{$D=D_{j-1}$ does not defend against $A$ at Line \ref{lin:GreedyIf}.}}
In this case, we have $D_j = D_{j-1} + j'$ where $j'=\max (N[A] \setminus D_{j-1})$.
Note that $j'$ exists since $N[A] \setminus D_{j-1}$ is non-empty.
Indeed, there is at least one vertex of $A \subseteq N[A]$ that is not in $D_{j-1}$, for otherwise $D_{j-1}$ defends against $A$.
Since $A$ is consecutive and $j'$ is chosen as the maximum among those vertices in $N[A]\setminus D_{j-1}$,
there are two mutually exclusive and complementing cases: either $j'\in N[j]$ or else $j'\in A \setminus N[j]$. 
We now show that in both cases, $D_j$ defends against $A$.
\begin{itemize}
\item {\underline{Case 1:} $j'\in N[j]$}.  
In this case $D_j$ defends against $A$ as follows.
The vertex $j'$ defends $j$ (possibly, $j'=j$), and $D_{j-1}$ defends against $A-j$ by the inductive hypothesis.
\item{\underline{Case 2:} $j'\in A\setminus N[j]$}.
Since $j' = \max (N[A] \setminus D_{j-1})$, we have that $[j'+1 .. \max (N[A])] \subseteq D_{j-1}$.
Let $D_{j-1}^+ = [j'+1 .. \max (N[A])]$, and $D_{j-1}^- = D_{j-1} \setminus D_{j-1}^+$.
By the inductive hypothesis, $D_{j-1}$ defends against $A-j$. 
It follows by Lemma \ref{lem:defenseMonotonic} that there is a monotonic defense of $D_{j-1}$ against $A-j$; 
$D_{j-1}^+$ defends the last $\abs{D_{j-1}^+}$ vertices of $A-j$, and $D_{j-1}^-$ defends the remaining vertices of $A-j$. 
As the last $\abs{D_{j-1}^+}$ vertices of $A-j$ form a consecutive attack and $D_{j-1}^+ = [j'+1 .. \max (N[A])]$ defends it, 
by Lemma \ref{lem:defenseShift}, $[j' .. \max (N[A])-1] \subseteq D_j$ defends against the last $\abs{D_{j-1}^+}$ vertices of $A-j$.
Clearly, $\max (N[A]) \in N[j]$. Then $\max (N[A])$ defends $j$.
Therefore, $D_j$ defends against $A$.
\end{itemize}
\end{itemize}

It remains to show (ii),
{\it i.e.}, that $D_j$ is contained in some optimal solution.
By the inductive hypothesis, there exists an optimal solution $D^* \supseteq D_{j-1}$ of {\kdefprobl}.
Suppose that $D_j \nsubseteq D^*$.
Then, $j' \notin D^*$.
Since $D_{j-1}$ does not defend against $A$, there is at least one vertex $j'' \in (D^* \setminus D_{j-1}) \cap N[A] \subseteq N[A] \setminus D_{j-1}$. 
By the way $j'$ is chosen by the algorithm, we have $j'' \leq j'$.
Let $\bar D^*=D^* - j'' + j'$. 
Clearly, $\abs{\bar D^*}=\abs{D^*}$ and $D_j \subseteq \bar D^*$. 
We now show that $\bar{D}^*$ is a $k$-defensive set.

By Lemma \ref{lem:consecutive}, it is sufficient to show that $\bar{D}^*$ defends against every consecutive $k$-attack.
Suppose that there exists a consecutive $k$-attack $\bar{A}$ that $\bar{D}^*$ does not defend against.
Let $\bar{j} = \max (\bar{A})$.
Putting together the facts that $D^*$ defends against $\bar{A}$, $\bar{D}^* \supseteq D_j$ defends every $k$-attack on $V_j$, and $D^*$ and $\bar{D}^*$ differ exactly by $j'$ and $j''$,
it must be the case that a) $\bar{j} > j$, and b) $j'' \in N[\bar{A}]$, $j' \notin N[\bar{A}]$.
Then $j' \in N[A] \setminus N[\bar{A}]$.
Therefore, the interval $I_{j'}$ does not intersect $\uu(\bar{A})$. 
Moreover, $I_{j'}$ is on the left of $\uu(\bar{A})$ since it intersects $\uu(A)$.
As $j'' < j$ the same holds for $I_{j''}$.
This contradicts the fact $j'' \in N[\bar{A}]$.

We conclude the proof by analyzing the time complexity of $\alggreedy$. If $D$ is maintained as a sorted linked list of vertices, the dominant part of the main loop becomes Line \ref{lin:GreedyIf} that checks whether $D$ defends against $A$,
and this can be implemented in time $\bigoh(k)$ using Lemma \ref{lem:defenseMonotonic}. 
\end{proof}

\section{An Improved Greedy Algorithm Using the Bubble Model}\label{sec:bubblesModel}
In this section, we present the bubble model for proper interval graphs (introduced in \cite{Heggernes2007TR,HMP09}) and modify Algorithm {\alggreedy} of the previous section to provide a better time complexity by taking advantage of this model.

A \emph{2-dimensional bubble structure} $\bb$ for a finite non-empty set $V$ is a 2-dimensional arrangement of bubbles $\set{B_{i,j}~|~j \in [1..c], i \in [1..r_j]}$ for some positive integers $c, r_1, \ldots r_c$, such that $\bb$ is a \emph{near-partition} of $V$. That is, $V = \cup \bb$ and the sets $B_{i,j}$ are pairwise disjoint, allowing for the possibility of $B_{i,j}=\emptyset$ for arbitrarily many pairs $i,j$. For an element $v \in V$, we denote by $i(v)$ and $j(v)$ the unique indices such that $v \in B_{i(v),j(v)}$.

Given a bubble structure $\bb$, the graph $G(\bb)$ defined by $\bb$ is the following graph:
\begin{enumerate}
\item $V(G(\bb))=\cup \bb$, and
\item $uv \in E(G(\bb))$ if and only if one of the following holds:
\begin{itemize}
\item $j(u)=j(v)$,
\item $j(u)=j(v)+1$ and $i(u) < i(v)$.
\end{itemize}
\end{enumerate}

We say that $\bb$ is a \emph{bubble model} for $G(\bb)$.
A \emph{compact representation} for a bubble model is a list of \emph{columns} each of which contains a list of non-empty bubbles, and each bubble contains its row number in addition to the number of vertices in this bubble. 
We summarize the main result about bubble representations of proper interval graphs as follows.

\begin{theorem}\cite{Heggernes2007TR,HMP09}\label{thm:Bubbles}
\begin{enumerate}
\item A graph is proper interval if and only if it has a bubble model.
\item A bubble model for a graph on $n$ vertices contains $O(n^2)$ bubbles and it can be computed in $O(n^2)$ time.
\item Given a proper interval graph on $n$ vertices, a bubble model of its compact representation can be computed in $O(n)$ time.
\end{enumerate}
\end{theorem}

It follows from the definition of a bubble representation that every column represents a clique, i.e., a set of pairwise adjacent vertices, and two vertices of a bubble are twins.
For simplicity we assume an arbitrary order on the set of vertices of a bubble.
A bubble model implies a so-called \emph{nested neighborhood} structure between any two consecutive columns of a bubble representation. The set of vertices in two consecutive columns induces a co-chain graph. 
As such, proper interval graphs can be seen as a succession of co-chain graphs  \cite{HMP09}, generalizing the notion of co-chain graphs.

A \emph{linear bubble representation} is a list of bubbles ordered by their column order of the bubble representation and by their row order within every column. 
In other words, vertex 1 and all its twins are in $B_1$, the first vertex not in $B_1$ and all its twins are in $B_2$, and so on. 
We add an artificial bubble $B_0=\emptyset$.
See Figure \ref{fig:ProperIntervalRep2} for an example.
The following lemma shows how we can obtain information necessary to develop an efficient greedy algorithm in linear time from a compact representation. 
We note that this information can be derived from the proof of Theorem 4 in \cite{Heggernes2007TR}; here we give it in a concise way for the sake of completeness.

\begin{figure}[h]
\begin{center}
\scalebox{1}{\commentfig{
\includegraphics[width=11cm]{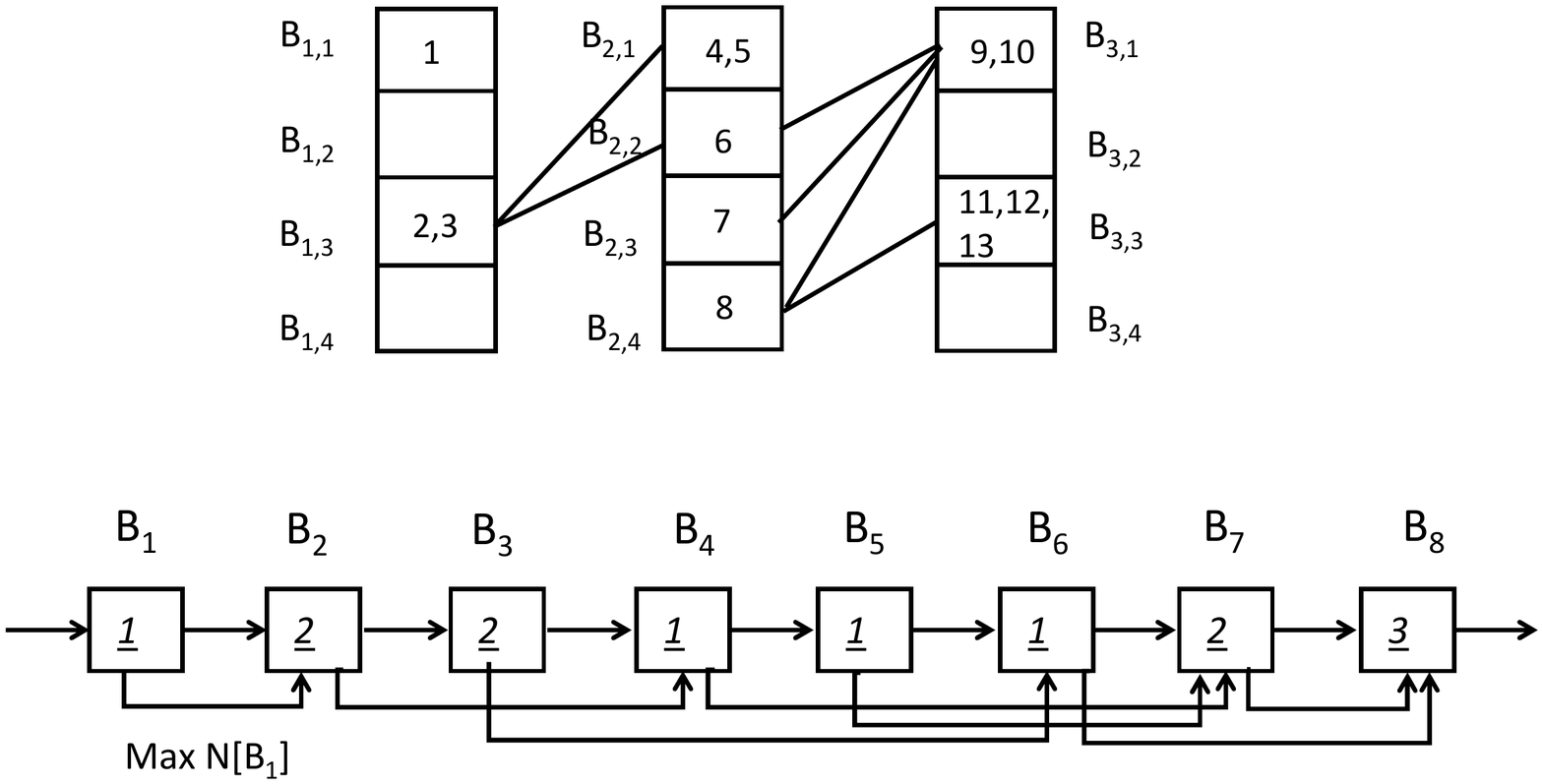}}}
\end{center}
\vspace{-1.5cm}
\caption{A bubble model of the proper interval graph in Figure \ref{fig:ProperIntervalRep1} and the corresponding linear bubble representation depicted below. Note that the vertex set of the graph can be partitioned into three cliques $\set{1,2,3},\set{4,5,6,7,8},\set{9,10,11,12,13}$. Every such clique is represented by a column of the bubble representation. The inclined lines joining two consecutive columns depict the adjacencies implied by the definition of a bubble model. The linear bubble model contains 8 bubbles, as the number of non-empty bubbles in the two dimensional bubble representation. The number written in the square representing a bubble is the number of vertices in it. Every bubble points to the next bubble. The elbow arrow leaving bubble $B_i$ points to the bubble that contains $\max(N[B_i])$.} \label{fig:ProperIntervalRep2}

\end{figure}

\begin{lemma}
A linear bubble representation along with $\abs{B}, \max(B), \max(N[B]), \min(B)$, and  $\min (N[B])$ for every bubble $B \neq \emptyset$ can be computed from a compact bubble representation $\bb$ in time $\bigoh(\abs{\bb})$.
\end{lemma}
\begin{proof}
We first note that for every two consecutive columns $j$ and $j+1$ and for every non-empty bubble $B_{i,j}$, we can determine a pointer $pMax$ to the last bubble adjacent to $B_{i,j}$ by processing the bubbles of column $j$ and column $j+1$ in decreasing row order of each. 
We start with the last non-empty bubbles of columns $j$ and $j+1$ as the current bubbles;
then skip bubbles of column $j+1$ until reaching a non-empty bubble of column $j+1$ that is adjacent to the current bubble of column $j$, which becomes the value of $pMax$ for the bubble of column $j$. 
We then move to the previous non-empty bubble of column $j$ and repeat the above process.
The time complexity of this procedure is the sum of the number of bubbles in columns $j$ and $j+1$.
Note that adjacencies can be determined in constant time by using the row and column indices of the bubbles.
By applying this procedure to every two consecutive columns, we can compute the $pMax$ pointers for all the bubbles in time $\bigoh(\abs{\bb})$.

After this stage, we concatenate the column lists into one list and re-index the bubbles by making the first bubble of the first column $B_1$ and so on.
We then compute $\max(B_i) = \sum_{j=1}^i \abs{B_j}$ for each bubble $B_i$.
Each of these steps can be clearly performed in time $\bigoh(\bb)$. 
Finally, we observe that $\max(N[B_i])$ is equal to $\max(B_j)$, where $B_j$ is the bubble pointed by the pointer $pMax$ of $B_i$; $\min(B_i) = \max (B_{i-1})+1$, and $\min (N[B_i])$ can be handled symmetrically.
\end{proof}

The algorithm presented in this section simulates algorithm {\alggreedy}, but instead of considering vertices one at a time, it processes them in larger chunks. 
Roughly speaking, it considers a bubble at a time in the order implied by a linear bubble representation.
The following lemma enables us to simulate  multiple iterations of {\alggreedy} in a single iteration as long as the right endpoint of $A$ remains in the same bubble.

\begin{lemma}\label{lem:rightbubble}
Let $A$ be a $k$-attack considered at some iteration of {\alggreedy} and $D$ be the set computed at that iteration to defend $A$.
Let $A'$ be a $k$-attack considered in a subsequent iteration with set $D'$ computed to defend $A'$.
If the last bubbles of $A$ and $A'$ in a linear bubble representation are the same, then $D' \setminus D \subseteq N[A] \setminus D$.
\end{lemma}

\begin{proof}
Suppose that the claim is not correct and assume without loss of generality that the first iteration of {\alggreedy} that adds an element $x \notin N[A] \setminus D$ is the one that considers the $k$-attack $A'$. 
Since $\max (A)$ and $\max (A')$ are in the same bubble they are twins, and $\min (A) \leq \min (A')$ we have $N[A'] \subseteq N[A]$.
By the behavior of {\alggreedy}, we have $x \in N[A'] \setminus (D' - x) \subseteq N[A] \setminus D$.
Then $x \in N[A'] \setminus N[A] = \emptyset$, a contradiction. 
\end{proof}

A monotonic defense $f$ of $D$ against a consecutive attack $A=[i..j]$ is \emph{rightmost} if $f(\max (D))=j$ and  every prior element $j' < j$ of $A$ is defended by the rightmost possible element of $D$ that is not used to defend an element $i\prime\prime \in [j'+1..j]$. 

We say that a bubble $B_i \in \bb$ is a bubble of $f$ if $f(v)$ is defined for at least one
element $v \in B_i \cap D$.
For a rightmost defense $f$ of $D$ against a consecutive attack $A$, and a bubble $B_i$ of $f$ 
we denote $f_i = f(\max (B_i))$ and $slack_i(f) = \max (N[B_i]) - f_i$.
Intuitively, $slack_i(f)$ is the potential of an element $u$ of $B_i \cap D$ to defend elements to the right of $f(u)$. 
We define the overall slack of a defense $f$ as $slack(f)=\min slack_i(f)$ where the minimum is taken over all bubbles of $f$.
This value determines the potential of the current set of defenders to defend elements to the right of $A$.

For a consecutive attack $A=[i..j]$ and an integer $\delta$ we denote by $A+\delta$ the ``shifted'' consecutive attack $[i+\delta..j+\delta]$. 
For a rightmost defense $f$ against a consecutive attack $A$, and an integer $\delta$, we denote by $f'=f+\delta$ the defense obtained by shifting $f$ by $\delta$ positions, i.e., $f'(v)=f(v)+\delta$, provided that $f(v)$ is defined and $v$ is adjacent to $f(v)+\delta$, which it is for all $v$ when $\delta \leq slack(f)$.

\newcommand{\alggreedyBubble}{\textsc{GreedyByBubble}}
\newcommand{\addNewVertices}{\textsc{AddNewVertices}}
\newcommand{\zeroSlackVertex}{\textsc{Bottleneck}}
\newcommand{\removeLeft}{\textsc{removeLeft}}

This leads us to algorithm {\alggreedyBubble}, given in Algorithm $\ref{alg:GreedyBubble}$ which maintains the invariant that $D$ admits a rightmost  monotonic defence $f$ against $A$ at the end of every iteration.
The algorithm employs several functions. 
The function $\addNewVertices(\Delta, A, D, f)$ adds $\Delta$ vertices to $A$ and $\Delta$ vertices to $D$ and updates $f$ so that the vertices added to $D$ defend the vertices added to $A$.
The function $\zeroSlackVertex(f)$ is called when $slack(f)$ is zero, and returns the index of the rightmost vertex of $A$ defended by a zero-slack bubble of $f$.  
The function $\removeLeft(\delta, A, f)$ removes the leftmost $\delta$ vertices of the (current) attack $A$ and defense $f$.

\alglanguage{pseudocode}

\begin{algorithm}
\caption{\alggreedyBubble}\label{alg:GreedyBubble}
\begin{algorithmic}[1]
\Require {A proper interval graph $G=([1..n],E)$.}
\Require {An integer $k$ less than $n$.}
\Ensure {Return a $k$-defensive subset $D$ of $[1..n]$ of smallest cardinality.}
\Statex
\State $\bb \gets$ a linear bubble representation of $G$. \label{lin:computeBubble}
\State $f \gets D \gets \emptyset$
\LineComment{$A = [firstIndex .. lastIndex]$.}
\State $firstIndex = 1$; $lastIndex = 0$;
\Comment {$A \gets \emptyset$}
\State \Call {\addNewVertices} {$k$, $A$, $D$, $f$}
\While {$lastIndex <= n$}
  	\If {$slack(f) > 0$}
		\State $A \gets A+slack(f)$ \label{lin:shiftA}
		\State $f \gets f + slack(f)$ \label{lin:shiftF}
	\Else  
		\State $v \gets $ \Call {\zeroSlackVertex}{$f$} \label{lin:findZeroSlack}
		\Comment The rightmost vertex of $A$ defended by a zero-slack bubble of $f$
		\State $move \gets \min \set{n-lastIndex, v-firstIndex+1}$
		\State \Call {\removeLeft} {$move$, $A$, $f$}  \label{lin:removeLeft}
		\State \Call{\addNewVertices} {$move$, $A$, $D$, $f$} \label{lin:addNewVertices}
	\EndIf
\EndWhile
\State \Return $D$
\Statex
\Function{\addNewVertices}{$\Delta$, $A$, $D$, $f$}
    \While {$\Delta > 0$}
        \State $i \gets $ the last bubble of $A$
        \If {$\max (B_i) = lastIndex$}
        \State $\delta \gets 1$
        \Else 
        \State $\delta \gets \min \set{\Delta, \max (B_i) - lastIndex}$
        \EndIf
        \State $lastIndex = lastIndex + \delta$ \label{lin:shiftLastIndex}
        \State $\Delta \gets \Delta - \delta$
        \State $NEW \gets $ the $\delta$ rightmost vertices of $N[A] \setminus D$
        \State $D \gets D \cup NEW$
        \State Update $f$ so that the vertices of $NEW$ defend the last $\delta$ vertices of $A$ \label{lin:updateF}
    \EndWhile
\EndFunction
\end{algorithmic}
\end{algorithm}

\begin{theorem}
Algorithm {\alggreedyBubble}  returns a minimum $k$-defensive set $D$ of a given proper interval graph $G$.
\end{theorem}

\begin{proof}
The proof is based on the correctness of Algorithm  {\alggreedy}.   
{\alggreedy} goes vertex by vertex through $G$.  
At each iteration, it attempts to “shift” the current defense to the right by one vertex to defend against the new attack (i.e., the last $k-1$ vertices of the prior attack and one new vertex).  
If it can defend this new attack, then nothing is done.  If it can not, then it adds one vertex to $D$, being the rightmost neighbor of the new attack that is not already in $D$.

Algorithm {\alggreedyBubble} maintains a rightmost defense $f$ of $A$ at every iteration. 
This invariant is initially vacuously true and incrementally maintained by Line \ref{lin:updateF}.

When $slack(f)>0$, the defense $f+1$ defends against the attack $A+1$ for the next $slack(f)$ iterations, and
Algorithm  {\alggreedy} would not add new vertices to $D$.  
As such, Algorithm  {\alggreedyBubble}, immediately updates the attack $A$ and the defense $f$ $slack(f)$ places to the right.

However, when $slack(f) = 0$, let $B_i$ be the rightmost bubble of $f$ with zero slack, and let $v=f(\max (B_i))=\max (N[B_i])$.
In other words, $v$ is the rightmost vertex of $A$ defended by a zero-slack bubble of $f$ (since $f$ is monotonic).
None of the vertices in the current defense $f$ that defend vertex $v$ and any prior vertex of $A$ can be used to defend vertices beyond $v$.
Thus, Algorithm  {\alggreedy} would have to add one new vertex to $D$ at every iteration until $v \notin A$.
The number  of these iterations is the number of vertices in $A$ from its leftmost vertex to $v$.
These iterations are simulated by a single invocation of 
 {\addNewVertices} that shifts $A$ by this number of vertices, computes the vertices to be added to $D$ and updates $f$ accordingly.
Function {\addNewVertices} achieves its goal by moving the right endpoint of $A$ to the last vertex of the current bubble and then to the first vertex of the next bubble in an alternating manner.
In the latter case it simply performs one step of {\alggreedy}.
In the former case it takes advantage of Lemma \ref{lem:rightbubble} that enables us to compute the vertices to be added to $D$ by just considering the current attack $A$ (since the rightmost vertex of $A$ remains in the same bubble).

Given the correctness of Algorithm  {\alggreedy}, Algorithm  {\alggreedyBubble} is correct, as it is equivalent in its actions.
\end{proof}

The following theorem establishes the time complexity of our algorithm.
\begin{theorem}
Algorithm {\alggreedyBubble}  can be implemented so as to return an optimal solution of {\kdefprobl}
in time $\bigoh(n + \abs{\bb} \cdot \log k)$.
\end{theorem}

\begin{proof}
We first describe the data structures that the algorithm maintains.
We add to every bubble $B_i \in \bb$ an additional value $d_i$, where $d_i = \abs{D \cap B_i}$.
We maintain four doubly-linked lists of these bubbles; one for each of $\bb$, $D$, $\bb \setminus D$ and $f$.
As for the set $A$, in addition to the integers $firstIndex$, $lastIndex$ explicitly mentioned in the algorithm, 
we maintain two pointers $firstBubble$ and $lastBubble$ that point to the first and last bubble of $A$, respectively.

The bubbles $B_i \in \bb$ that are used in the defense $f$ are also maintained in a min-Heap $F$ supported by an additional integer $offset$.
The key of a bubble $B_i$ in the heap $F$ is an ordered pair $(key_i, -i)$ such that the invariant $slack_i(f) = key_i - offset$ is maintained.
The items in the heap and the elements of the linked list of $f$ point to each other.

Since every bubble enters at most once and leaves at most once each of the four doubly-linked lists, the number of these operations each of which takes constant time is at most $8 \abs{\bb}$.
Furthermore, since $\abs{F} \leq k$ every insertion into and deletion from the heap $F$ takes $\bigoh(\log k)$ time; the number of such operations is at most $2 \abs{\bb}$.
Therefore, the total time spent to update these data structures is $\bigoh(\abs{\bb} \cdot \log k)$.

We now show that these data structures are sufficient to implement the algorithm and that the rest of the operations take $\bigoh(\abs{\bb})$ time.
For a bubble $B_i$ in $f$ we have $key_i = slack_i(f) + offset$.
Then, $f_i = \max (N[B_i]) + slack_i(f) = \max (N[B_i]) + key_i - offset$ which can be computed in constant time.
Moreover, $slack(f)= \min_{i \in F} key_i - offset$ can be computed in constant time by inspecting the first element of the heap.
Similarly, the function {\zeroSlackVertex} can be implemented in constant time since the first element of the heap contains the rightmost bubble of $f$ with minimum (in this case zero) slack.
When $slack(f) > 0$, the shifting of $f$ at line \ref{lin:shiftF} can be done in constant time by simply incrementing $offset$ by $slack(f)$.

We now consider the number of iterations of the algorithm.
We observe that between every two iterations of the main loop in which $slack(f)>0$, at least one bubble must have been added to $f$.
Therefore, the number of iterations in which $slack(f)>0$ is $\bigoh(\abs{\bb})$.
On the other hand, when $slack(f)=0$, except possibly for the last iteration, at least one bubble is removed from $f$. 
Thus the number of iterations in which $slack(f)=0$ is $\bigoh(\abs{\bb})$.  

Thus,  the overall time complexity is $\bigoh(n + \abs{\bb} \cdot \log k)$ where the term $n$ accounts for the computation of the linear bubble model.
\end{proof}

\section{Concluding Remarks}
In this work, we answer the open question about the complexity of {\kdefprobl} in proper interval graphs by providing a linear-time (polynomial) algorithm. 
It should be noted that the algorithm presented here is not a direct extension of the algorithm designed for co-chain graphs. 
Specifically, our algorithm considers a constant number of attacks per bubble on every instance.
On the other hand, the algorithm in \cite{EFP} for co-chain graphs is based on the idea that it is sufficient to consider only two ``worst" attacks in order to determine a minimum $k$-defensive dominating set in linear time. 
Whether or not it suffices to consider a smaller number of attacks for proper interval graphs, (for instance, a constant number of attacks per column of the bubble representation) is an open question.

Another open research direction would be to identify other graph classes where {\kdefprobl} is polynomial-time solvable. 
More precisely, one could investigate {\kdefprobl} in interval graphs, or graph classes with (not necessarily nested but) structured neighborhoods. 
Such structured neighborhoods are expressed by the existence of special vertex orderings; as such, chain graphs, bipartite permutation graphs and biconvex graphs are good candidates for investigation. 
They are all subclasses of (general) bipartite graphs for which the complexity of {\kdefprobl} is unknown, even for fixed $k$. 
In the same spirit, one could consider $t$-trees defined as graphs that can be obtained from a clique on $t$ vertices by repetitively adding a new vertex and making it adjacent to all vertices of a clique on $t$ vertices in the current graph. 
Finally, we would like to point out that there is no clear relationship between the complexities of {\kdefprobl} and {\kplusdefprobl} although one would intuitively expect that when $k$ gets larger the problem is at least as difficult as for smaller $k$ values.

\section*{\refname}
\bibliographystyle{abbrv}
\bibliography{GraphTheory,References}

\begin{thebibliography}{1}

\bibitem{Bogart199921}
K.~P. Bogart and D.~B. West.
\newblock A short proof that ‘proper = unit’.
\newblock {\em Discrete Mathematics}, 201(1–3):21 -- 23, 1999.

\bibitem{kra}
D.~Dereniowski, T.~Gavenciak, and J.~Kratochvil.
\newblock Conditional graph-theory iv: Dominating sets.
\newblock {\em Theoretical Computer Science}, 794:47--58, 2019.

\bibitem{EFP}
T.~Ekim, A.~Farley, and A.~Proskurowski.
\newblock The complexity of the defensive domination problem in special graph
  classes.
\newblock {\em Discrete Mathematics}, 343:111665, 2020.

\bibitem{FP}
A.~Farley and A.~Proskurowski.
\newblock Defensive domination.
\newblock {\em Congressus Numerantium}, 168:97--107, 2004.

\bibitem{HH1}
F.~Harary and T.~Haynes.
\newblock Conditional graph-theory iv: Dominating sets.
\newblock {\em Utilitas Mathematica}, 48:179--192, 1995.

\bibitem{HH}
T.~Haynes, S.~Hedetniemi, and P.~Slater.
\newblock {\em Fundamentals of Domination in Graphs}.
\newblock Marcel Dekker, New York, 1998.

\bibitem{Heggernes2007TR}
P.~Heggernes, D.~Meister, and C.~Papadopoulos.
\newblock A new representation of proper interval graphs with an application to
  clique-width.
\newblock Technical Report 354, Department of Informatics, University of
  Bergen, Norway, 2007.

\bibitem{HMP09}
P.~Heggernes, D.~Meister, and C.~Papadopoulos.
\newblock A new representation of proper interval graphs with an application to
  clique-width.
\newblock {\em Electronic Notes in Discrete Mathematics}, 32:27--34, 2009.

\end{thebibliography}

\end{document}